\documentclass[11pt, letterpaper]{article}
\usepackage[utf8]{inputenc}
\usepackage{fullpage}
\usepackage{bm}
\usepackage{hyperref}
\usepackage{amsmath, amssymb, amsthm}
\usepackage{color}
\usepackage{cleveref}

\newcommand{\expect}[2]{\mathbb{E}_{#1}\left[#2\right]}

\newcommand{\prob}[1]{\Pr\left[#1\right]}

\newtheorem{definition}{Definition}

\newtheorem{lemma}{Lemma}
\newtheorem{proposition}{Proposition}
\newtheorem{remark}{Remark}

\title{Would Friedman Burn your Tokens?}

\author{
  Aggelos Kiayias\\
  University of Edinburgh, IOG\\
  \texttt{akiayias@inf.ed.ac.uk}
  \and
  Philip Lazos\\
  IOG\\
  \texttt{philip.lazos@iohk.io}
  \and
  Jan Christoph Schlegel\\
City, University of London\\
  \texttt{jansc@alumni.ethz.ch}
}

\date{\today}

\begin{document}

\maketitle

\begin{abstract}
Cryptocurrencies come with a variety of {\em tokenomic} policies as well as aspirations of desirable monetary characteristics that have been described by proponents as ``sound money'' or even ``ultra sound money.'' These propositions are typically devoid of economic  analysis so it is a pertinent question how such aspirations fit in the wider context of monetary economic theory. In this work, we develop a framework that determines the optimal token supply policy of a cryptocurrency, as well as investigate how such policy may be algorithmically implemented. Our findings suggest that the optimal policy complies with the Friedman rule and it is dependent on the risk free rate, as well as the growth of the cryptocurrency platform. Furthermore, we demonstrate a wide set of conditions under which such policy can be implemented via contractions and expansions of token supply that can be realized algorithmically with block rewards,  taxation of consumption and burning the proceeds, and blockchain oracles.
\end{abstract}

\section{Introduction}

Tokenomics, referring to the algorithmic adjustment of token supply in blockchain systems, is one of the most contested aspects of cryptocurrencies. The most prominent example, Bitcoin, adopts a fixed total supply monetary policy with a progressively diminishing release schedule that follows a geometric distribution and converges to a supply of 21 million tokens~\cite{nakamoto2008}. This feature of bitcoin has been touted as one of its main advantages --- namely the fact that the protocol enforces a fixed total supply (in conjunction with its other features)  has been suggested to imply ``sound money''~\cite{bitcoin-standard}. Contrary to this, the second most popular cryptocurrency, Ethereum, adopts an ever growing supply that can turn deflationary depending on usage --- this is the case due to the fee burning mechanism adopted through the EIP-1559 upgrade in the Ethereum system. This policy has been advocated as ``ultra-sound money''  by Ethereum proponents.\footnote{See e.g., \url{https://ultrasound.money/}.}

Despite the variety of mechanisms proposed, there has been little 
analysis of cryptocurrency monetary policies.
It is thus of interest to investigate how such mechanisms stack up against conventional monetary theory. 
%
One of the foundational results of pure monetary theory, is the celebrated Friedman rule~\cite{friedman69} which posits that monetary supply should be adjusted to keep the nominal interest rate close to zero. 
Following this tenet, the optimal money supply should seek a rate of inflation or deflation equal to the real {\em risk free} return rate in the economy.
 While theoretically appealing, such policy is rarely observed in reality and a variety of prior work investigate the conditions under which it is optimal and implementable 
 \cite{Wilson1979,CoKo1998,Ireland2002}.

Viewing cryptocurrencies through the lens of the Friedman rule, we ask what is the optimal token supply policy when these systems are viewed in the context of wider economic activity and furthermore how such optimal policy can be implemented algorithmically -- i.e., without relying on a central bank that controls the money supply.  This is of particular interest in this context, since monetary supply adjustments seem  to be hard to implement in practice in the context of token economies. More specifically,  expanding the token supply is easy to implement, for example through issuing block rewards, while on the other hand, a decrease in the token supply, can be more complicated to implement. 

One policy that has been popular is to tax transactions and to burn the tax proceeds, as suggested in the Ethereum blockchain after the EIP-1559 upgrade.\footnote{We note that the Ethereum base fee as handled in EIP-1559 can be interpreted as a transaction tax.}
Some alternative mechanisms, possible in theory, could be to decrease the token supply by maintaining an on-chain treasury with some part of the treasury being held in different assets. Given such treasury, the protocol can buy back tokens and hence reduce the money supply. In practice, this poses the challenge of funding such treasury as well as
adjusting the composition of its assets as prices change. 
Another alternative is levying a tax on token holding. However, if this policy is implemented uniformly and non-discriminatory (which is usually the only feasible implementation of such policy), this is just a form of re-basing that changes the unit of account without having any real economic effect.
 Taxing transactions has the potential downside that they could reduce user (and validator) welfare because they make transactions more expensive compared to actual competitive pricing. On the other hand, an upside of consumption based taxation in case of congestion is that it may satisfy the dual function of managing congestion as well as reducing the token supply.

\subsection{Our Contributions}
In this paper we set out to develop a framework that can determine the optimal token supply policy when a cryptocurrency is viewed in the context of the wider economic activity available to agents. We also focus on ways such optimal policy can be implemented using a taxation of consumption approach as the basic implementation mechanism. 
The higher objective is to give a proper theoretical foundation to a class of tokenomics policies for blockchain systems that have been discussed in the cryptocurrency community such as taxation and fee burning but only heuristically motivated, to the best of our knowledge.

We propose a fully micro-founded model of the token economy in the spirit of models popular in the micro-foundation of liquidity and payments (see the textbook-\cite{textbook} for an exposition of these type of models). In our model, users can use tokens to pay for transactions fees in the blockchain system which gives a natural demand for tokens. There is also an external (centralized) market where users and validators/miners can exchange tokens for fiat money and vice versa. We assume that there are frictions that prevent agents to transact in both markets at the same time. Moreover, we assume that expenditure on transactions is a small part of users' overall consumption expenditure and that they receive their income externally, so that we can model their choices by quasi-linear preferences in the numéraire good (usually we think of dollars as numéraire). This allows for a tractable model of the token economy that can be solved using standard techniques. As in other models of monetary economies there is a plurality of equilibria that depends on the agents' beliefs about the future value of the token. There is an equilibrium in which agents believe that tokens will not have value in the future, as well as equilibria in which they do. We focus on steady state equilibria in which agents believe that and hence the value of the token stock \emph{does} grow proportionally to the transaction volume in the platform.  

As a first result, we derive the robust finding that in the (steady state) equilibrium of the token economy the Friedman rule is optimal: We should target a token supply so that the expected return on token holding equals the risk-free rate which we can introduce into the cryptocurrency platform via a platform governance mechanism or a blockchain oracle~\cite{DBLP:conf/ccs/ZhangMMGJ20}. Such policy can however not always be (optimally) implemented. If a reduction of the token supply needs to be implemented through taxing transactions and the system is not congested such taxation is suboptimal. Not taxing and leaving the token supply unchanged would outperform the deflationary policy. However, if the system is congested, then tax and burn can be optimal. We derive precisely the conditions under which such policy can be optimal, approximates the Friedman rule, and derive the optimal taxation rate. 

For this purpose, we enhance our model by introducing preference shocks, similarly as in \cite{prat2021}. Preference shocks lead to uncertainty of the usefulness of tokens for consumption in the next period. If agents are homogeneously affected by shocks, reducing the token supply in the future through tax and burn schemes, has no welfare improving effect. However, if agents are heterogeneously affected by shocks, and there is congestion, due to increased demand of positively affected agents, then taxing consumption in that state has a positive effect on consumption in other states, thus enhancing overall welfare.

\subsection{Related work}
There has been a lot of recent work on economic aspects of cryptocurrencies under the umbrella term "tokenomics". Many researchers focus on the role of tokens as a medium of exchange. The token is the main medium of exchange in the economy~\cite{pagnotta2022} or competing with an alternative medium of exchange~\cite{garratt2018,schilling2019}. In the former case, the feedback between different aspects of the system such as system security~\cite{pagnotta2022,budish2022} and equilibrium prices is studied. In the latter case, variants of the exchange rate indeterminacy result of~\cite{kareken1981} hold. In contract to this previous work, our model is more in line with models of utility tokens where tokens have a particular role in the ecosystem~\cite{prat2021,Cong2021,CONG2022} rather than functioning as a general medium of exchange. Our model assumes that economic activity involving the token is a small part of overall economic activity (for example users would usually not receive their income in the token but in dollar and blockchain transactions is a small part of their overall consumption expenditure). Thus, our model resembles more these models of utility tokens. However, our work differs by studying the question of optimal token supply. The closest work in terms of the overall research question might be \cite{Cong2021}, who however focuses on the revenue of the platform operator not on social welfare. The set up of our model also resembles the work of \cite{prat2021} and \cite{pagnotta2022} who study a discrete model with trading frictions and shocks in the spirit of new monetarist models \cite{textbook}.

Our work also relates to work on the micro foundation of money. As observed by~\cite{kocherlakota1998}, a model where money (tokens) is essentially must satisfy several properties, among them trading frictions, limited commitment and the absence of a public record keeping technology. Researchers have provided a variety of models in this spirit, in particular~\cite{lagos2005} provide a tractable model of a monetary economy, by alternating a centralized market with a decentralized market. The centralized market (together with quasi-linearity of consumption) guarantees that the money distribution re-balances. We use the alternating market set up with the convention that the centralized market is the dollar-token market, which gives a natural interpretation in our application.

The EIP-1559 is the best-known (and most thoroughly studied) transaction fee mechanism which attempts to do both demand estimation and also burns parts of the transaction fees. The foundations of its study were laid by Roughgarden in \cite{roughgarden2020transaction}, who (among other results) identified that transaction fee burn is necessary to prevent certain types of `attacks', such as validators including their own fake traffic to inflate the prices. The efficacy of this particular demand estimation has been studied in \cite{leonardos2021dynamical, reijsbergen2021transaction}, where it was shown that despite the possibility of chaotic behaviour on the prices, the actual target block size is consistently achieved. In addition to EIP-1559, alternative mechanisms have been proposed (such as \cite{ferreira2021dynamic, chung2023foundations, kiayias2023tiered}) focusing on improving specific properties, such as stability, truthfulness or welfare.


\section{Model}
We consider a Lagos Wright~\cite{lagos2005} type of model: Time is discrete and infinite $t=0,1,\ldots$. There are two types of agents, users and validators. There is a unit mass of users and each user $i\in[0,1]$ has a per period utility $u_{i,t}(a, \sigma)$ of consuming a fraction $a$ of the blockspace\footnote{We are assuming that this is a non-atomic game: a single user with $a = 1$ would not consume the entire blockspace (and in fact would have no effect on other users, no matter her individual $a$) given a utility shock $\sigma_t$. However, a larger mass of users would consume a measurable fraction of the blockspace, and in fact may not be able to consume all of it if the cost of validation exceeds their utility}. In the following, we usually introduce time dependency of utility by assuming that utility grows deterministically, so that we can write $u_{i,t}(a, \sigma)=A_tu_{i}(a, \sigma)$ for functions $\{u_i\}_{i\in [0,1]}$ and positive numbers $\{A_t\}_{t=0,1,\ldots}$. There is a unit mass of identical validators, who can chose to accommodate an amount of $s$ total user activity for cost $c(s)$. We assume that all $u_i$ are concave and $c$ is convex.

The timing is as follows: At the beginning of each period $t$, users receive a shock $\sigma_t \sim F$, choose their activity level $a_{i,t}$ and pay for their activity a per unit fee $p_t$ to the validators. Payments for platform usage can only be settled in tokens. At the end of each period, users and validators decide on their token holding $m_t$ which are used in the next period $t+1$. Therefore, if a user wants to participate with level $a_i$ in platform usage in period $t$ she needs to hold at least $m_{i,t-1}$ dollar worth of tokens at the end of period $t-1$ such that $p_{t}a_i=(1+r_t^T)m_{i,t-1}$ where $r_t^T$ is the return on token holding between period $t-1$ and $t$. We assume that there is a competitive market at the end of the period at which users and validators can exchange tokens for the numéraire (in the following we use "dollar" as name for the numéraire) and vice versa. We assume that the activity on the platform is a small part of the total economic activity of the users (e.g. they receive their income in dollars, do most of their consumption in dollars etc.). In other words we do a partial equilibrium analysis: preferences are quasi-linear in dollars and the agents in the economy can hold arbitrary (possibly negative) dollar balances. Users discount time with the same discounting rate $\beta=1/(1+r)$ where $r$ is the real (riskless) return. \\

\noindent
For convenience, we summarize the definitions of all commonly used symbols in Table~\ref{table:notation}.
\begin{table}
\centering
\begin{tabular}{|c | c|}
    \hline
        Symbol & Definition \\ [0.5ex] 
    \hline\hline
        $u_i$ & utility of user $i$\\ 
    \hline
        $c$ & cost of validation\\
    \hline
        $a_{i, t}$ & fraction of the blockspace used by user $i$ during time $t$\\
    \hline
        $s_{j, t}$ & on-chain activity validated by validator $j$ during time $t$\\
    \hline
        $m_{i, t-1}$ & token holdings (in dollars) of user $i$, purchased at time $t-1$ to use at time $t$\\
    \hline 
        $p_t$ & per unit fee for usage at time $t$\\
    \hline
        $\theta_t$ & percentage tax on top of price $p_t$ that is burned at time $t$\\
    \hline
        $r$ & risk-free rate\\
    \hline
        $r_t^T$ & return on token holding at time $t$\\
    \hline
        $\beta$ & time discount factor, equal to $1 / (1+r)$\\
    \hline
        $q_t$ & dollar exchange rate of tokens at time $t$, such that $r_t^T = (q_t - q_{t-1}) / q_{t-1}$\\
    \hline
        $A_t$ & technology parameter at time $t$\\
    \hline
        $\sigma_t$ & utility shock received by the users at the beginning of time $t$\\
    \hline
\end{tabular}
\caption{An overview of the symbols used. When an index $i$ (or $j$) is missing from a symbol, we are integrating over all users or validators (e.g., $a_t = \int_0^1 a_{i, t} di$)\label{table:notation}}
\end{table}

\subsection{Equilibrium}

The utility over time for user $i$ with token holdings $m_{i, t}$ is given by
\begin{equation}\label{def:user_utility}
    \sum_{t=1}^{\infty}\beta^t 
        \expect{\sigma}{
            \max_{a_{i, t}\geq0, p_t a_{i, t}\leq (1+r_t^T)m_{i, t-1}}
            \left(u_{i,t}(a_{i, t}, \sigma_t)+(1+r_t^T)m_{i, t-1} - p_t a_{i, t}\right) - m_{i, t - 1} / \beta
            }.
\end{equation}

For readability, we will drop the subscript $i$ from $a_{i, t}$ and $m_{i, t}$ when it is clear from context. Note that given the token holding $m_{i, t-1}$, the user can determine her activity level after the shock has been revealed.





For validator $j$, we correspondingly obtain the expression
\begin{equation}\label{def:utility_validators}
\sum_{t=1}^{\infty}\beta^t \expect{\sigma_t}{
    \max_{s_{j, t}}
    \left(p_ts_{j, t}-c(s_{j, t})+(1+r_t^T)m_{j, t-1}\right)-m_{j, t - 1} / \beta},
\end{equation}
with the added restriction that $s_t = \int s_{j, t} dj \le \int a_{i, t} di$.
It is immediate to see that if the return on token holding $r_t^T$ is smaller or equal the real return $r$, then validators will choose to hold zero token balances at the end of each period (selling all tokens obtained as revenue from validation), $m_{j,t}=0$ for all $j$ and $t$. While the validator cost is not directly affected by the shock, it depends on the activity $a_t$ which is.

\subsubsection{First best}\label{sec:first_best}
As a benchmark, it is instructive to characterize the socially optimal allocation. The optimal level of participation is obtained by maximizing within each period for each realization of the shock:
\[\max_{a_{i,t}\geq0,i\in[0,1]}\int_0^1 u_{i,t}(a_{i,t}, \sigma_t)di-c\left(\int_0^1 a_{i,t}di\right)\]
which gives the first order conditions
\begin{equation}\label{eq:first_best_deterministic}
    \frac{u_{i,t}'(a_{i,t}, \sigma_t)}{c'(\int a_{i,t}di)}=1\quad\text{for each }i\in[0,1]
\end{equation}
In case that in the critical point we have $\int_0^1a_{i,t}di>1$, we can optimally ration the blockspace usage $\hat{a}_{i, t} \ge 0$ such that $\int_0^1 \hat{a}_{i,t}di = 1$ and $ u'_{i, t}(\hat{a}_{i, t}, \sigma_t) = C > 0$ for all $i$\footnote{This may not be possible for some sets of $u_{i, t}$ (e.g., if some utility is constantly zero). In this case, these users have $\hat{a}_{i, t} = 0$ and the condition holds for the remainder.}. 
\subsubsection{Competitive Equilibrium}
In the following we denote by $a^*_{i,t}(p, m, \sigma)$ the blockspace demand of user $i$ in period $t$, i.e. the solution to the optimal consumption problem (which is essentially one term of the sum in \ref{def:user_utility})
\begin{equation}\label{eq:user_opt_activity}
\max_{a\geq0,pa\leq m} u_{i,t}(a, \sigma) - p a
\end{equation}
given a fee of $p$, token holdings of $m$ and shock $\sigma$.
Moreover, we denote by  $\{m^*_{i,t}(p,r_t^T)\}_{t=0,1,\ldots}$ the token demand of user $i$, i.e. the solution to the problem
\begin{equation}\label{eq:user_opt_mit}
\max_{\{m_{i,t}\}_{t=0,1,\ldots}}
    \sum_{t=1}^{\infty}\beta^t 
    \expect{\sigma_t}{
        \left(u_{i,t}(a_{i, t}^*)+(1+r_t^T)m_{i, t-1}-p_t a_{i, t}^*\right)-m_{i, t - 1} / \beta
    },
\end{equation}
where we omitted the arguments from $a_{i, t}^*((1+r_t^T)m_{i, t-1}, p_t, \sigma_t)$ for readability. We will often do so when they are clear from the context.
The side of validators is simple: from \ref{def:utility_validators} to maximize their utility for a given price, they just need to accommodate enough activity at time $t$ such that $c'(s_t) = p_t$ (otherwise $s_t = 1$, since they can only accommodate up to this much activity).

Putting it all together we have the following definition.

\begin{definition}
    An equilibrium is described by prices $\{p_t, q_t\}_{t=0,1,\ldots}$, where $q_t$ is the dollar exchange rate of tokens (such that $r_t^T=\frac{q_t-q_{t-1}}{q_t}$), such that the usage levels 
    $\{a_{i,t}^*\}_{i\in[0,1],t=0,1,\ldots}$,
    token balances measured in dollar terms 
    $\{m^*_{i,t}\}_{i\in\mathcal{N},t=0,1,\ldots}$ and token supplies (measured in number of tokens) 
    $\{M_t\}_{t=0,1,\ldots}$ satisfy the following:
\begin{gather*}
(c'(a^*_t)=p_t\text{ and }a^*_t<1)\text{ or } (c'(a^*_t) \le p_t \text{ and } a^*_t = 1)\\
a^*_t=\int_0^1 a^*_{i,t}\left(p_t,(1+r_t^T)m^*_{i,t-1}, \sigma_t\right)di\\ m^*_t=\int_0^1 m^*_{i,t}\left(p,r_t^T\right)di=q_tM_t.
\end{gather*}
\end{definition}
\begin{remark}
    The first condition has two cases, depending on the on-chain congestion. Either the price is such that validator supply and user demand are in equilibrium, with total usage less than 1 \emph{or} total usage is exactly equal to the blockchain throughput and validators validate all of it. Given that the utilities are concave and the costs convex, for a specific $q_t$ there may exist only one corresponding $p_t$ leading to an equilibrium, where one of the two cases of the first conditions are met.
\end{remark}

\subsubsection{Steady State Equilibrium Selection}
There are various token equilibria that are characterized by different beliefs about the dollar value of the token supply over time, $\{q_tM_t\}_{t=0,1,\ldots}$, e.g. there is always an equilibrium where all agents believe that tokens have no value.\footnote{While there are equilibria in which tokens have no value, in our model all equilibria in the dollar-token market will be such that the return on token holding $r_t^T$ is at most the real risk-less return $r$. This is an artifact of the removal of wealth effect and borrowing constraints through quasi-linear preferences in dollar holdings: assuming quasi-linearity, the agents are risk neutral. Thus, if agents would expect $r_t^T>r$, agents would demand an infinite amount of tokens. However, excess expected return on token holding, could be captured in an extended version of our model, e.g. by adding speculators in a similar way as \cite{mayer2021}.} Equilibria where the dollar value of the token supply is positive and grows with the technology parameter (which scales all user utilities) seem realistic.
\begin{definition}[Steady State Equilibrium]
\[m_t = q_tM_t = \frac{A_t}{A_{t-1}} q_{t-1}M_{t-1}.\]
In that case \[r_t^T=\frac{q_t-q_{t-1}}{q_{t-1}}=\frac{A_t}{A_{t-1}}\cdot\frac{M_{t-1}}{M_{t}}-1.\]
\end{definition}
Notice that we have not added the impact of transaction fee burning to the token supply $M_t$. The equilibrium definition only depends on the final $M_t$, which can be affected by different policies.

\subsubsection{Simplifying the Notation}
In the first sections, where the proofs are generally simpler, we follow the notation fairly closely, with only minor deviations when the meaning is absolutely clear from context. However in later sections, once some familiarity has been established with the indices and arguments to commonly used functions, we will increasingly omit them for readability and only re-introduce them to draw particular attention to them.

\section{Optimality of the Friedman rule}
We begin by showing that the Friedman rule from monetary policy (i.e., setting the rate of inflation (or deflation) of $M_t$ to be equal to the risk-free rate) is also optimal for this token economy. We present a version of this result in two parts:
\begin{itemize}
    \item First for a deterministic setting, including technological growth.
    \item Then, we complement this result with a more general version that also incorporates  uncertainty about the users' utilities in the form of the shocks $\sigma_t$.
\end{itemize}

We assume that the technology parameter grows deterministically with a rate $\gamma$ so that
\[A_{t}=(1+\gamma)A_{t-1}.\]
which implies
\[u_{i,t}(a_t)=(1+\gamma)^t u_i(a_t)\]






\begin{proposition}\label{prop:friedman}
For a steady state equilibrium, the optimal per-period change in the token supply is given by
\[\frac{M_t}{M_{t-1}}=\frac{1+\gamma}{1+r}\]
where $\gamma$ is the usage growth and $r$ is the risk-free rate. If the token supply is optimal, then the equilibrium rate of return on token holding is equal to the risk-free rate, $r_t^T=r.$
\end{proposition}
\begin{proof}
We first consider the case where the chain at time $t$ is un-congested in the first best \ref{eq:first_best_deterministic}:
\begin{equation}\label{eq:friedman_first_order}
\frac{A_t u_{i,t}'(1)}{c'(1)}<1\text{ for all }i\in[0,1]
\end{equation}
In this case, there exist optimal $\hat{a}_{i, t}$ such that
\[\frac{\int_0^1 A_t u_{i,t}'(\hat{a}_{i,t})di}{c'(\int_0^1\hat{a}_{i,t}di)}=1.\]
If there are no shocks or demand uncertainty, and assuming that the return on token holding is smaller than the return on the alternative asset $r_t^T\leq r$ it should be intuitively clear that it is optimal to choose token holdings so that everything is spent in the consecutive period:
\begin{equation}\label{eq:alpha_mit_connection}
(1+r_t^T)m_{i,{t-1}}=p_ta_{i,t}.
\end{equation}
This can be shown by combining \ref{eq:user_opt_activity} and \ref{eq:user_opt_mit}: if $(1+r_t^T)m_{i,{t-1}} > p_ta_{i,t}$ then either $a_{i, t}$ can be increased or $m_{i, t-1}$ decreased, staying within the feasible region and improving the outcome for user $i$.

To determine the optimal token holding we derive the following first order condition, which represents the choice of $m_{i, t-1}$ that user $i$ would have to make to maximize \ref{eq:user_opt_mit}:
\[
-1 + \beta\left(A_t u_{i,t}'(a_t^*)\tfrac{\partial a_t^*}{\partial m_{i,t-1}}+(1+r_t^T)-p_t\tfrac{\partial a_t^*}{\partial m_{i,t-1}}\right)=0
\]
We can rewrite this as:
\begin{equation}\label{eq:user_first_friedman}
A_t u'_i(a^*_{i,t})\tfrac{1+r_t^T}{p_t}+(1+r_t^T)-(1+r_t^T) 
= 
1/\beta = 1+r
\Rightarrow
\frac{A_t u'_i(a^*_{i,t})}{p_t}
=
\frac{1+r}{1+r_t^T}=1+i_t,
\end{equation}
where \[1+i_t\equiv \frac{1+r}{1+r_t^T}\] is the nominal return on the alternative investment and we have used \ref{eq:alpha_mit_connection} to obtain $\tfrac{\partial a_t^*}{\partial m_{i,t-1}} = \tfrac{1+r_t^T}{p_t}.$
Integrating over all users (where $i \in [0,1]$) and using the optimality condition for validators, $c'(a_t^*)=p_t$ we get
\[\frac{\int A_t u'_i(a^*_{i,t})di}{c'(a_t^*)}=\frac{1+r}{1+r_t^T}=1+i_t.\]
For optimal consumption, monetary policy should target $i_t=0$ (or equivalently $r_t^T=r$) so that $\hat{a_t} = \int a_{i, t}^* di= a_t^*$ from \ref{eq:friedman_first_order} and given concavity and convexity of utility and cost respectively.
Focusing again on the steady state equilibrium, we have
\[r_t^T=\frac{q_t-q_{t-1}}{q_{t-1}}=\frac{A_t}{A_{t-1}}\cdot\frac{M_{t-1}}{M_{t}}-1=(1+\gamma)\cdot\frac{M_{t-1}}{M_{t}}-1.\]
Thus the optimal policy would target a change of the token supply given by
\[1+r=(1+\gamma)\cdot\frac{M_{t-1}}{M_{t}}\Rightarrow\frac{M_t}{M_{t-1}}=\frac{1+\gamma}{1+r}.\]
The token stock should extend if $\gamma>r$ and shrink if $\gamma<r$.

Next we consider the case where the chain at time $t$ is congested in the first best \ref{eq:first_best_deterministic}:
\[
    \frac{A_t u_{i,t}'(1)}{c'(1)}>1\text{ for }i\in[0,1]
\]
As we have shown in \ref{sec:first_best}, the blockspace has to be rationed so that the total usage is exactly 1. To do this, the price $p_t$ needs to be used to limit congestion and not just to cover the costs of validating. Clearly, since for every user $u'_{i, t}(1) > c'(1)$, we would need that $p_t > c'(1)$ in order to tame the demand for the optimal outcome. Therefore, the validators will all set $s_{j, t} = 1$.

The users are faced with the identical optimization problem, whose solution satisfies Equation \ref{eq:user_first_friedman} as before:
\[
\frac{A_t u'_i(a^*_{i,t})}{p_t}
=
\frac{1+r}{1+r_t^T}
\Rightarrow
A_t u'_i(a^*_{i,t})
=
p_t \frac{1+r}{1+r_t^T}.
\]
Setting $r_t^T = r$ allows $p_t$ to be the unique solution to $\int a^*_{i, t}((1 + r_t^T)m_{i, t-1}, p_t, \sigma_t)di = 1$, which as explained satisfies $p_t > c'(1)$ and is identical to the first best solution.\\

These two cases can be easily combined for the general case, where some users have $u_{i,t}'(1) > c'(1)$ and other $u_{i,t}'(1) \le c'(1)$, leading to the same result about $r_t^T = r$.
\end{proof}

We extend the results on the optimality of the Friedman rule to the case of preference shocks. For simplicity, we assume that the profile of utility functions and the cost function satisfy the following non-degeneracy condition: The optimal consumption vector $(a_{i,t}(\sigma))_{i\in[0,1],\sigma\in\Sigma}$ in each equilibrium for each state is different
\[
    \frac{u_{i,t}'(a_{i,t}( \sigma),\sigma)}{c'(\int a_{i,t}(\sigma)di)}
\neq \frac{u_{i,t}'(a_{i,t}( \sigma'), \sigma')}{c'(\int a_{i,t}(\sigma')di)}\text{ for }\sigma\neq \sigma'.
\]
Now the per period utility of consumption is given by
\[u_{i,t}(a, \sigma)\]
where $\sigma$ is the preference shock. The user chooses optimal $m_{t-1}$ and $a_t(\sigma)$:
\[-m_{t-1}+\beta \expect{\sigma}{u_{i,t}(a_{t}, \sigma)-p_{t}(\sigma)a_{t}+(1+r_t^T)m_{t-1}}\]
subject to the constraint that
$pa_t\leq(1+r_t^T)m_{t-1}.$ We use $\pi_\sigma$ for the probability of getting a particular shock $\sigma$.
Taking first order conditions to derive the optimal $m$ we get:
\begin{align*}
&-1+\beta\left\{1+\expect{\sigma}{r_t^T}+\pi_{\tilde{\sigma}}(u'(a_{i,t},\tilde{\sigma})-p_t{\tilde{\sigma}})
\tfrac{\partial a_{i, t}}{\partial m_{t-a}}\right\}\\
=
-1+\beta&\left\{1+\expect{\sigma}{r_t^T}+\pi_{\tilde{\sigma}}\frac{(u'(a_{i,t}, \tilde{\sigma})-p_t(\tilde{\sigma}))(1+r_r^T(\tilde{\sigma}))}{p_t(\tilde{\sigma})}\right\}=0
\end{align*}
where $\tilde{\sigma}$ is the unique state where the budget constraint is binding in the optimum (notice that the budget constraint is only binding in that state, since for all other states $\sigma' \neq \tilde{\sigma}$ the budget is not exhausted, therefore $\tfrac{\partial a_{i,t}(\sigma')}{\partial m_{t-1}} = 0$ and they disappear from the first order condition).
Rearranging
\[
\frac{u'(a_{i,t}(\tilde{\sigma},m))}{p}=1+\frac{r-\expect{\tilde{\sigma}}{r_t^T}}{\pi(\tilde{\sigma})(1+r_t^T(\tilde{\sigma}))}
\]
It is again optimal to target $\mathbb{E}_{\sigma}r_t^T = r$ and we have established the optimality of the Friedman rule.\\

\noindent
The general case, where the budget constraint is binding in multiple states follows along similar lines.

\section{Implementation for Deterministic Demand}
As we observed during the proof of \Cref{prop:friedman}, the token supply $M_t$ should shrink (or extend) based on the risk-free rate $r$ and technological progress rate $\gamma$. However, this result only tells us by how much $M_t$ should change at the steady state equilibrium, not how to actually achieve this. In an actual implementation, we need to decide which users to impact when reducing $M_t$ or which to support when increasing it. How this is done could have a detrimental effect (e.g., by reducing consumption). In this section we focus on deterministic demand, where the shock is irrelevant for all utilities $u_i$ and show that, while increasing the token supply is fairly easy and could be achieved through a variety of means, decreasing it is trickier. In fact, if we intend to do so by taxing consumption and burning the excess we should not hope to achieve an improved outcome and may actually hurt the welfare of the users.

\subsection{Increasing Token Supply}
An increase in the token supply, which would be warranted if $\gamma > r$, can be implemented in several ways. Agents in our model have quasi-linear preferences in the numéraire good and are not borrowing constrained. Thus, it has no effect on total welfare to whom additional tokens are distributed, but only a distributional effect. For overall welfare, it is therefore irrelevant (in the model) whether tokens are distributed to validators, users or other stake holders. A natural way to implement token issuance is through block rewards that validators would receive in addition to the transaction fees.

\subsection{Decreasing Token Supply through Taxing Consumption and Burning}
While increasing the token supply is easy to implement, the same cannot be said for the decreasing it. Our primary goal is not to reduce $M_t$ per se, but to increase the return on token holdings $r_t^T$. However, the two are very closely related in the steady state equilibria we consider. There are many policies that could accomplish this, such as treasury operations, transaction fee redistribution or transaction fee burning. We will focus on the latter. Given that $p_t$ is necessary to compensate the validators, we can `tax' usage by $\theta_t$ and burn the proceeds. As such, the users would need to pay $a_{i, t} (1+\theta_t) p_t$ for activity $a_{i,t}$, but only $p_t$ would go to the validators. The rest $\theta_t p_t$ would be burned, reducing the token supply accordingly. The notion of equilibrium previously introduced ca be readily adapted to the case of taxation, with the only modification that validators and users now face different prices, $p_t$ respectively $(1+\theta_t)p_t$. When studying the case where the only policy that affects token supply is fee burning, we add a constraint of the form $M_t = M_{t-1} - \theta_t p_t a_t / q_t$ to the steady state equilibrium.

As we observed during the proof of \Cref{prop:friedman}, token supply should shrink (or extend) based on the risk-free rate $r$ and technological progress rate $\gamma$. For it to be optimal to decrease the token supply, it should be the case that usage grows less than the overall economic activity, $r>\gamma$. 

We summarize the results in the following theorem. Even though the effect of taxation does not lead to improved outcomes, the actual effect varies depending on the on-chain congestion.
\begin{proposition}
    For deterministic demand, decreasing token supply through transaction fee burning cannot lead to increased welfare. In particular, for $r > \gamma$:
    \begin{itemize}
        \item For low congestion  the first best is un-attainable. Taxation has a neutral effect that balances decreased demand due to fees with token appreciation through burning.
        \item For high congestion, the first best is attainable. Taxation still has a neutral effect and does not impact validators.
    \end{itemize}
\end{proposition}
\begin{proof}
We first consider the case where the chain is un-congested in the first best \ref{eq:first_best_deterministic}:
\[\frac{A_t u'_i(1)}{c'(1)}<1\text{ for each }i\in[0,1].\]
As previously observed, there exist optimal $a^*_i$ such that
\[\frac{\int A_t u_i'(a^*_{i,t})di}{c'(\int a^*_{i,t}di)}=1.\]

For a given tax rate $\theta_t$ the optimization problem faced by user $i$ changes compared to \ref{eq:user_opt_mit}. From that user's perspective, the fee becomes $(1 + \theta_t) p_t$:
\begin{equation*}
\max_{\{m_{i,t}\}_{t=0,1,\ldots}}
    \sum_{t=1}^{\infty}\beta^t 
    \expect{\sigma_t}{
        \left(A_t u_{i,t}(a_{i, t}^*)+(1+r_t^T)m_{i, t-1}-(1 + \theta_t)p_t a_{i, t}^*\right)-m_{i, t - 1} / \beta
    },
\end{equation*}
Following the same calculations as before, we can derive a first order condition for consumption
\begin{equation}\label{eq:first_nocongest_determininstic}
    \frac{\int A_t u'_i(a^*_{i,t})di}{(1+\theta_t)c'(a^*_t)}=\frac{1+r}{1+r_t^T}=1+i_t.
\end{equation}
Since we are in the un-congested case (and clearly the increased fees due to taxation can only reduce consumption), we still have that $a^*_t < 1$ and thus $p_t = c'(a^*_t)$ at equilibrium. 

Suppose all tax revenue is burned in each period:
\begin{align}
M_t = M_{t-1}-\theta_tp_ta_t/q_{t} &\Rightarrow
q_t M_t = q_t M_{t-1} - \theta_t p_t a_t\notag\\
&\Rightarrow m_t = \left(1 + \frac{q_{t} - q_{t-1}}{q_{t-1}}\right) q_{t-1} M_{t-1} - \theta_t p_t a_t\notag\\
&\Rightarrow m_t = (1+r_t^T)m_{t-1}-\theta_tp_ta_t\label{eq:txburn}
\end{align}
where $r_t^T=\tfrac{q_t-q_{t-1}}{q_{t-1}}$, $m_t=q_tM_t$ and $m_{t-1}=q_{t-1}M_{t-1}$.
Focusing again on steady state equilibria we have
\[
    m_t=q_tM_t=q_{t-1}M_{t-1}=(1+\gamma)m_{t-1},
\]
which we can combine with Equation~\ref{eq:txburn} to obtain:
\begin{equation}\label{eq:yield_burn}
    (r_t^T - \gamma)m_{t-1}=\theta_tp_ta_t.
\end{equation}
Using the fact that users spend all their tokens in each period,  $m_{i,t-1}=\tfrac{(1+\theta_t)p_ta_{i,t}}{1+r_t^T}$ and integrating over all users to get $m_{t-1}=\tfrac{(1+\theta_t)p_ta_{t}}{1+r_t^T}$, we can substitute:
\[
(r_t^T - \gamma)\frac{(1+\theta_t)p_ta_{t}}{1+r_t^T}
=
\theta_tp_ta_t
\Rightarrow
\frac{r_t^T - \gamma}{1+r_r^T}
=
\frac{\theta_t}{1+\theta_t}
\Rightarrow
\theta_t
=
\frac{1 + r_t^T}{1 + \gamma} - 1.
\]
Plugging this back into Equation~\ref{eq:first_nocongest_determininstic}, the first order condition becomes
\begin{equation}\label{eq:det_demand_final}
    \frac{\int u'_i(a^*_{i,t})di}{(1+\theta_t)c'(a^*_t)}=\frac{1+r}{1+r_t^T} \Rightarrow
    \frac{\int u'_i(a^*_{i,t})di}{c'(a^*_t)} = \frac{1 + r}{1 + \gamma}.
\end{equation}
Therefore, taxation has a neutral effect, where the positive effect of token appreciation through token burning and the negative effect of demand reduction through consumption taxation exactly cancel out. Note that the outcome might be worse than the optimal, since $(1 + r)/(1 + \gamma) > 1$ reduces $a^*_{t}$ compared to the optimal and pushes it further into the region where the marginal of the utility exceeds the marginal of the cost.

In the case, of a congested chain in the first best
\[\frac{A_t u_i'(1)}{c'(1)}>1\text{ for }i\in[0,1],\]
consumption taxation still has an overall neutral effect on welfare, but its effect is qualitatively different. Taking the first order condition for users we obtain a similar result as Equation~\ref{eq:det_demand_final}, but $p_t$ is not necessarily equal to $c'(a^*_t)$ due the congestion. In addition, we do not integrate over all users, but consider the activity $a^*_{i, t}$ of an individual user to obtain:
\begin{equation}
    \frac{u'_i(a^*_{i,t})}{(1+\theta_t)p_t}=\frac{1+r}{1+r_t^T}
    \Rightarrow
    u'_i(a^*_{i,t}) = p_t \frac{1 + r}{1 + \gamma}.
\end{equation}
As before, notice that the choice of $\theta_t$ has a neutral effect on user activity, which is fully defined by the choice of $p_t.$ To limit congestion and reach an equilibrium, $p_t$ has to be set such that either 
\[
    \int a^*_{i, t} di = 1 \text{ and } p_t \ge c'(1)
\]
or
\[
    \int a^*_{i, t} di < 1 \text{ and } p_t = c'\left(\int a^*_{i, t} di\right),
\]
and only one of the two is possible. So even though the chain is congested in the first best, given that $(1+r)/(1+\gamma) > 1$ we may be unable to set the price $p_t$ high enough to satisfy both $\int a^*_{i, t} di = 1$ and  $p_t \ge c'(1)$ for the first condition. If this is the case, then the chain is un-congested and we fall back to the previous case. If not, observe that that there exists a unique $p_t$ (independent of $\theta_t$) that leads to the optimal outcome. Therefore, $\theta_t > 0$ has a neutral effect for any congestion level.
\end{proof}
\begin{remark}
In the proof of this result, notice that the choice of $\theta_t$ did \emph{not} impact the validator payments (in dollars) in the congested case! Even though the fee $(1+\theta_t) p_t$ may aim to control congestion through a higher tax that is burned (thus intuitively suggesting that the validators are paid less), the increased yield $1 + r_t^T$ makes this effect neutral: the users have to purchase fewer tokens that appreciate before they have to spend them, but whatever remains after collecting the tax also appreciates equally after the validators collect it. In addition, notice that with high enough congestion there may be a price $p_t$ leading to the optimal outcome where the entire blockspace is used, even though there are carry costs involved. This cannot be achieved for deterministic demand in the un-congested setting (assuming $r > \gamma$).
\end{remark}

\section{Uncertain demand and shocks}
So far, it may seem that there is no hope to obtain a better outcome by increasing the return on token holdings $r_t^T$ through taxation and burning. However, this is not always the case. The previously studied cases had in common that the users had no uncertainty about how many tokens should be optimally held. Users were `uniformly' affected by our particular policy. To have a positive effect, we need to study a situation where users are \emph{uncertain} about how many tokens to purchase and may not be able to spend all of them in every state of the world. Precisely this can be modelled through random \emph{shocks} and increased congestion. Combining both, the users that were positively affected by the shock may spend all the coins and potentially `crowd-out' the rest who have leftover tokens that appreciate in value. This averages out (to an extent) the effects of the shocks and due to the concavity of the utilities and convexity of costs there is indeed a tangible improvement in the outcome, provided taxation $\theta_t$ is not too high. Taxation in some states of the world can increase utility in other states of the world by decreasing the carry cost of token holdings and hence increasing consumption closer to the optimal level.

As we will see in the following, taxation can be optimal if agents are heterogeneously affected by preference shocks \emph{and} a positive preference shock for a subgroup of agents leads to congestion. On the other hand, if agents are homogeneous, or if preference shocks don't lead to congestion, taxation does not increase welfare.

\paragraph*{Binary Shocks}
We consider the case where $\sigma_t \in \{0, 1\}$ is binary and $\prob{\sigma_t = 1}$ is known to the users. In general, $\sigma_t = 1$ will be referred to as the high demand (or positive) shock and the utilities will be higher, while $\sigma_t = 1$ will result in lower demand.

\subsection{Homogeneous agents ex-ante}
We first begin by considering homogeneous agents, to check if there are benefits to taxation stemming from `distributing' the gains of individual agents from their utility from the high demand shock state to the other.

Let us consider the simple where all agents have the same utility $u$ in case of a positive shock and $0$ otherwise.
We consider two cases. First, the one where shocks are independent for different users. Afterwards, the one where all users receive the same shock.

We denote by $\bar{r}_t^T$ the return on token holding in case of a positive shock $\sigma_t=1$  and by $\underline{r}_t^T$ the return on token holding in case of a no shock $\sigma_t=0$.  We denote $\bar{a}_{i, t}(m,p):=a^*_{i, t}(m, p, 1)$ and $\underline{a}_{i, t}(m, p):=a^*_{i, t}(m, p, 1)$, following the convention in \ref{eq:user_opt_activity}.

\subsubsection{Users Receive Independent Shocks}
In the first case, we consider users that receive independent and identically distributed binary shocks (this slightly abuses the notation, since shocks $\sigma_t$ are not defined per user). Let $\rho = \prob{\sigma_t = 1}$ be the fraction of users that receive a positive shock. Even though setting $\theta_t > 0$ still is at best neutral, the effect is different than the deterministic case.

\begin{proposition}\label{prop:iid_shocks}
    For binary shocks $\sigma_{t} \in \{0, 1\}$, received independently and identically distributed by every user, burning $\theta_t$ of the fees cannot lead to increased welfare.
\end{proposition}
\begin{proof}
Notice that on aggregate (due to having a continuum of users) the activity $a_t$, price $p_t$, token holdings $m_t$ and returns on token holdings $r_t^T$ are not random.  After the shocks realize, a fraction $\rho$ of users will choose $\bar{a}_{t}((1+r_t^T)m,p)\geq 0$ whereas a fraction $1-\rho$ will choose $\underline{a}_{i,t}=0$ (since their utility for $\sigma_t = 0$ is 0).
Since agents are the same ex-ante (which is when they decide their token holdings), the maximization they face is the same for all and it is the following (using $\bar{a}$ and omitting the arguments $(1+r_t^T) m_{t-1}$ and $(1+\theta_t) p_t$):
\begin{align*}
&\max_{m_{t-1}}
    \left\{
        -m_{t-1} + \beta \expect{\sigma}{u(a^*_t)+(1+r_t^T)m_{t-1} - (1 +\theta_t)p_t a^*_t}
    \right\}\\
=&\max_{m_{t-1}}
    \left\{
        -m_{t-1} + \beta\left(\rho (u(\bar{a}_t)+(1+r_t^T)m_{t-1}-(1+\theta_t)p_t\bar{a}_t)+(1-\rho)(1+r_t^T)m_{t-1}\right)
    \right\}\\
=&\max_{m_{t-1}}
    \left\{
        -m_{t-1}+\beta\left(\rho (u(\bar{a}_t)-(1+\theta_t)p_t\bar{a}_t+(1+r_t^T)m_{t-1}\right)
    \right\}
\end{align*}
Taking first order conditions:
\[
\beta\left(\rho (u'(\bar{a}_t)-(1+\theta_t)p_t)\tfrac{\partial \bar{a}_t}{\partial m_{t-1}}+1+r_t^T\right)=1
\]
If the return on token holding is smaller than risk free return, then the agent will optimally choose $m$ such that all of it is spent in the high state:
$(1+r_t^T)m_{t-1}=(1+\theta_t)p_t\bar{a}_t$. 
Therefore: $\frac{\partial \bar{a}_t}{\partial m_{t-1}}=\frac{1+r_t^T}{(1+\theta_t)p_t}$ and:
\[
    \beta(\rho(u'(\bar{a})-(1+\theta_t)p_t)\tfrac{\partial \bar{a}}{\partial m_{t-1}}+1+r_t^T)
    =
    \beta(1+r_t^T)\left(1+\rho\left(\frac{u'(\bar{a})}{(1+\theta_t)p_t}-1\right)\right)=1.
\]
Given that $\beta=\tfrac{1}{1+r}$ and defining
\[1+i_t:=\frac{1+r}{1+r_t^T}\] we therefore get
\[\frac{u'(\bar{a})}{(1+\theta_t)p_t}=1+\frac{i_t}{\rho}.\]

The effect of taxation is different than in the deterministic case. \Cref{eq:yield_burn} still applies for this case and gives $(r_t^T - \gamma)=\theta_tp_t\rho \bar{a}((1+r_t^T)m_{t-1}, p)/m_{t-1}$. Combining this with the condition that all money is spent in the high state, we get
\begin{align*}
1 + r_t^T & =1 + \gamma + \theta_tp_t\rho \bar{a}((1+r_t^T)m_{t-1}, p)/m_{t-1}\\
&= 1 + \gamma + \rho \frac{\theta_t (1+ r_t^T)}{1 + \theta_t},
\end{align*}
and finally solving for $1 + r_t^T$:
\[
    1+r_t^T=\frac{(1 + \gamma)(1+\theta_t)}{1+(1-\rho)\theta_t}
\]
Combining with the first order condition:
\begin{align*}
\frac{u'(\bar{a}_t)}{p_t}
&=
(1+\theta_t)\left(1 + \frac{1 + i_t - 1}{\rho}\right)\\
&=
(1+\theta_t)\left(1 + \frac{\frac{1 + r}{1 + r_t^T} - 1}{\rho}\right)\\
&=
(1+\theta_t)\left(
1 + 
\frac{1}{\rho}\cdot
\frac{1 + r - (1+\gamma)(1+\theta_t)+(1+r)(1-\rho)\theta_t}{(1+\gamma)(1+\theta_t)}\right)\\
&=
1+\frac{(r-\gamma)(1+(1-\rho)\theta_t)}{\rho}.
\end{align*}
Notice that the right hand side is greater than or equal to 1 and \emph{increasing} in $\theta_t$. As always, there are two options depending on congestion.

In the un-congested state, we have that $p_t = c'(\rho \bar{a}_t)$ and the right hand side should be equal to 1 in an optimal outcome. Thus we need to set taxes $\theta_t$ such that the distortion is minimal, and the right hand side as close to $1$ as possible, which is achieved for $\theta_t = 0.$ In a congested state, we just need that the total activity $\bar{\rho a} = 1$. This may be achievable for many $p_t, \theta_t$ combinations (unlike the deterministic case!). However, there is no improvement for $\theta_t > 0$, only the risk that the price $p_t$ would have to be set so low that the validators might not process the entire activity $\rho \bar{a}$.
\end{proof}
\begin{remark}
    Contrary to the deterministic case, there may be multiple choices of $p_t$ and $\theta_t$ leading to an optimal outcome in the congested setting (instead of just one choice of $p_t$ and irrelevance of $\theta_t$). The difference lies in the shock: the users who receive the low shock do not use their tokens and benefit more from the fee burn, whereas in the deterministic case all non-burned tokens would be transferred to validators.
\end{remark}

\subsubsection{All Users Receive the Same Shock}
Typically, the shocks received by each user would be correlated to each other and perhaps guide a greater market trend, given the new information presented. To model this, we consider the case where all agents receive the same shock with probability $\rho$ for the the positive shock. Compared to \Cref{prop:iid_shocks}, the return on token holdings $r_t^T$ is not deterministic anymore, since the aggregate demand is not deterministic anymore. As before, we consider identical users that have utility $u$ for the positive shock and $0$ otherwise.

\begin{proposition}
    For binary shocks $\sigma_t \in \{0, 1\}$, where each user receives the same shock, burning $\theta_t$ of the fees cannot lead to increased welfare.
\end{proposition}
\begin{proof}
Given the two states of the world (based on the shock), we define the corresponding variables $\bar{p}_t$, $\bar{\theta}_t$ and $\bar{r}_t^T$. Conveniently, given that for the negative shock $u_i = 0$, we have that the corresponding variables are $\bar{p}_t = \bar{\theta}_t =\bar{r}_t^T = 0$ as well.

As before, if the return on token holding is smaller than the risk free rate, then the agent will optimally choose $m$ such that all of it is spent in the good state: $(1+\bar{r}_t^T)m_{t+1}=(1+\bar{\theta}_t)\bar{p}_t\bar{a}_t(m_{t-1}, \bar{p}_t)$. Combining this with the condition from \Cref{eq:yield_burn} that $(\bar{r}_t^T - \gamma) m_{t-1}=\bar{\theta}_t \bar{p}_t \bar{a}_t$, we get
\begin{align*}
    1 + \bar{r}_t^T 
    = 
    (1+\bar{\theta}_t)\bar{p}_t\bar{a}_t / m_{t-1}
    =
    (1+\bar{\theta}_t)\frac{\bar{r}_t^T - \gamma}{\bar{\theta}_t}
    =
    (1 + \bar{\theta}_t)(1 + \gamma)
\end{align*}
By construction, we have no consumption and zero return in the bad state so that
\[
    \expect{\sigma}{1 + r_t^T}=(1 - \rho) + \rho(1 + \bar{\theta}_t)(1 + \gamma).
\]
Thus the user solves
\begin{align*}
&\max_{m_{t-1}}\left\{
-m_{t-1} + \beta \expect{\sigma_t}{u(a^*_t) - (1+\theta_t)p_ta^*_t+(1+r_t^T)m_{t-1}}\right\}\\
=&
\max_{m_{t-1}} \left\{
-m + \beta(\rho (u(\bar{a}_t)-(1+\bar{\theta}_t)\bar{p}_t \bar{a}_t)+
((1 - \rho) + \rho(1 + \bar{\theta}_t)(1 + \gamma))m_{t-1})
\right\}
\end{align*}
Taking first order conditions:
$$
\rho (u'(\bar{a}_t)-(1+\bar{\theta}_t)\bar{p}_t)\tfrac{\partial \bar{a}_t}{\partial m_{t-1}} + 
(1 - \rho) + \rho(1 + \bar{\theta}_t)(1 + \gamma)
= 
1 / \beta
$$
Since all of $m_{t-1}$ is spent on the positive shock state (i.e. $(1+\bar{r}_t^T)m_{t-1}=(1+\bar{\theta}_t)\bar{p}_t\bar{a}_t$) we have that $\frac{\partial \bar{a}_t}{\partial m_{t-1}}=\frac{1+\bar{r}_t^T}{(1+\bar{\theta}_t)\bar{p}_t}$ and therefore:
\begin{align*}
&\rho (u'(\bar{a}_t)-(1+\bar{\theta}_t)\bar{p}_t)\tfrac{\partial \bar{a}_t}{\partial m_{t-1}} + 
(1 - \rho) + \rho(1 + \bar{\theta}_t)(1 + \gamma)
= 
1 / \beta\\
\qquad\Rightarrow&
(u'(\bar{a}_t)-(1+\bar{\theta}_t)\bar{p}_t))\tfrac{\partial \bar{a}_t}{\partial m_{t-1}}
= 
\frac{1 + r - (1 - \rho) - \rho(1 + \bar{\theta}_t)(1 + \gamma)}{\rho}\\
\qquad\Rightarrow&
u'(\bar{a}_t)\frac{(1+\bar{\theta}_t)(1+\gamma)}{(1 + \bar{\theta}_t) \bar{p}_t}
=
\frac{1 + r - (1 - \rho) - \rho(1 + \bar{\theta}_t)(1 + \gamma)}{\rho}
+
(1+\bar{\theta}_t)(1+\gamma)\\
\qquad\Rightarrow&
\frac{u'(\bar{a}_t)}{\bar{p}_t}
=
\frac{1}{1+\gamma}\frac{1 + r - (1 - \rho) - \rho(1 + \bar{\theta}_t)(1 + \gamma)}{\rho}
+
1 + \bar{\theta}_t\\
\qquad\Rightarrow&
\frac{u'(\bar{a}_t)}{\bar{p}_t}
=
\frac{\rho + r}{(1+\gamma)\rho}. 
\end{align*}
In conclusion, the choice of $\theta_t$ does not matter (either with or without congestion), as long as $\expect{\sigma}{r_t^T} \leq r_t.$
\end{proof}
\begin{remark}
    It seems that either for independent or identical shocks, it is not possible to relieve the carry cost (thus increasing token holdings and possible welfare) in general, by taxing consumption. It is unclear if this extends for homogeneous ex-ante users that have a non-zero utility in the low-congestion state, but it seems unlikely, since the token holdings would be determined by the positive shock state for all users.
\end{remark}

\subsection{Heterogeneous agents}
We modify the previous examples, assuming that the agents are not homogeneous. In particular, there are two utility functions $\bar{u}$ and $\underline{u}$ (intuitively, $\bar{u}$ leads to higher demand than $\underline{u}$). There are two types of agents: a fraction $\lambda$ of agents  is of Type A and has utility $\bar{u}$ in case of a positive shock and $\underline{u}$ otherwise, while a fraction $1-\lambda$ is of Type B and has utility $\underline{u}$ in both states. Moreover, we assume that
\[
\frac{\bar{u}'(1 / \lambda)}{c'(1)} > 1 > \frac{\underline{u}'(1)}{c'(1)}.
\]
In the positive shock state, there is insufficient blockspace to optimally satisfy all users. In particular, the users with $\bar{u}$ alone are happy to consume the entirety of the blockspace, thus potentially crowding out the rest.

\begin{proposition}
    For binary shocks $\sigma_t \in \{0, 1\}$, heterogeneous users and high congestion in the positive shock state, there is a consumption taxation policy that leads to improved social welfare.
\end{proposition}
\begin{proof}
Given the two different user types, we need to define a few more symbols. Specifically we have use $\bar{x}$ and $\underline{x}$ to indicate that this particular variable $x$ is conditioned on the positive or negative shock. Specifically:
\begin{itemize}
    \item Type A users: $\bar{a}_t$ and $\underline{a}_t$ for the activity in the high and low congestion shock states and $\bar{m_t}$ for token holdings.
    \item Type B users: $\bar{b}_t$ and $\underline{b}_t$ for the activity and $\underline{m}$ for token holdings.
    \item Prices $\bar{p}_t, \underline{p}_t$, taxes $\bar{\theta}_t, \underline{\theta_t}$ and finally returns $\bar{r}_t^T, \underline{r}_t^T$. 
\end{itemize}
For simplicity, we assume that there is no technological progress and $\gamma = 0$, although the results generalize. In addition, we do not tax in the low congestion state (i.e., $\underline{\theta}_t = 0$ and thus $\underline{r}_t^T = 0$.

For the Type A user, we have the following optimization:
\begin{align*}
&\max_{\bar{m}_{t-1}} \left\{
- \bar{m}_{t-1} + \beta(\expect{\sigma_t}{(1 + r_t^T)\bar{m}_{t-1} + u(a^*_t) - (1+\theta_t) p_t a^*_t)}
\right\}\\
=
&\max_{\bar{m}_{t-1}} \{- \bar{m}_{t-1} +\\
&\qquad\beta((1 +\expect{\sigma_t}{r_t^T})\bar{m}_{t-1} +
(1-\rho)(\underline{u}(\underline{a}_t) - \underline{p}_t \underline{a}_t)
+ \rho(\bar{u}(\bar{a}_t) - (1 + \bar{\theta}_t) \bar{p}_t \bar{a}_t))
\}
\end{align*}
leading to the following first order condition:
$$
- 1 + \beta\left((1 + \expect{\sigma_t}{r_t^T}) + 
(1-\rho)(\underline{u}'(\underline{a}_t) - \underline{p}_t)\tfrac{\partial \underline{a}_t}{\partial \bar{m}_{t-1}}
+ \rho(\bar{u}'(\bar{a}_t) - (1 + \bar{\theta}_t) \bar{p}_t) \tfrac{\partial \bar{a}_t}{\partial \bar{m}_{t-1}} \right) = 0
$$
Since the Type A user only spends all tokens $\bar{m}_{t-1}$ in the high congestion shock, we have that $\tfrac{\partial \underline{a}_t}{\partial \bar{m}_{t-1}} = 0$. In addition, $(1 + \bar{r}_t^T)\bar{m}_t = (1 + \bar{\theta}_t)\bar{p}_t\bar{a}_t$, leading to $\tfrac{\partial \bar{a}_t}{\partial \bar{m}_{t-1}} = \tfrac{1 + \bar{r}_t^T}{\bar{p}_t(1 + \bar{\theta})}$. Plugging these back we get:

\begin{equation}\label{eq:good_tax1}
    (1 + \rho \bar{r}_t^T) +
    \rho(\bar{u}'(\bar{a}_t) - (1 + \bar{\theta}_t)\bar{p}_t \frac{1 + r_t^T}{(1 + \bar{\theta}_t)\bar{p}_t} = 1 / \beta
    \Rightarrow
    \frac{\bar{u}'(\bar{a}_t)}{(1 + \bar{\theta}_t) \bar{p}_t} = 
    \frac{1 + \frac{r}{\rho}}{1 + \bar{r}_t^T}
    = 1 + \frac{\frac{r}{\rho} - \bar{r}_t^T}{1 + \bar{r}_t^T}
\end{equation}

The Type 2 user solves a similar maximization, but the utility is deterministically $\underline{u}$:
\begin{align*}
    &\max_{\underline{m}_{t-1}} \left\{
    - \underline{m}_{t-1} + \beta(\expect{\sigma_t}{(1 + r_t^T)\underline{m}_{t-1} + \underline{u}(b^*_t) - (1+\theta_t) p_t b^*_t)}
    \right\}
\end{align*}

Before continuing, we show the following simple property of concave functions, that we will invoke later.
\begin{lemma}\label{lemma:poor_mans_convex}
Let $f, g : \mathbb{R} \rightarrow \mathbb{R}$ be smooth and strictly concave.
For $x_1 < x_2$ and $y_1 > y_2$, if $f'(x_1) > g'(y_1)$, $f'(x_2) \ge g'(y_2)$ and $\lambda x_1 + (1-\lambda) y_1 = \lambda x_2 + (1-\lambda) y_2 = 1$ then
\[
\lambda f(x_2) + (1-\lambda) g(y_2) \ge \lambda f(x_1) + (1-\lambda) g(y_1)
\]
\end{lemma}
\begin{proof}
    We take the derivative of $h(x) = \lambda f(x) + (1 - \lambda) g\left(\frac{1- \lambda x}{1 - \lambda}\right)$ with respect to $x$:
    \[
        h'(x) = \lambda f'(x) - \lambda g'\left(\frac{1- \lambda x}{1 - \lambda}\right).
    \]
    By our assumption, $h'(x_1) = \lambda(f'(x_1) - g'(y_1)) > 0$ and $h'(x_2) = \lambda(f'(x_2) - g'(y_2)) > 0$. Given $f, g$ are concave and smooth, $h'(x)$ is decreasing and continuous in $x$.
    Therefore $h'(x) > 0$ for $x_2 \ge x \ge x_1$. Integrating over that range:
    \[
        \int_{x_1}^{x_2} h'(x) \ge 0 \Rightarrow h(x_2) - h(x_1) \ge 0,
    \]
    concluding the proof.
\end{proof}

There are two cases for where the binding constraint on $\underline{m}_{t-1}$ is. Either the agent exhausts her tokens in the low state or in the high state. Since the two cases are similar, we only present the first case. The user will optimally choose $\underline{m}_{t-1}$ such that all of it is spent in the low state: $\underline{p}_t \underline{b}_t = \underline{m}_{t-1}$. Notice that the token holding constraint does not directly involve $r_t^T$ (although the price partly encodes this information) and that in the congested state she is crowded out and may not consume as much as she would like due to the limited block space.
Taking first order conditions and solving as before, we get:
\begin{equation}\label{eq:good_tax2}
    \frac{\underline{u}'(\underline{b}_t)}{\underline{p}_t} = 
    \frac{1 + r - \rho(1 + \bar{r}_t^T)}{1 - \rho}
    =
    1 + \frac{r - \rho\bar{r}_t^T}{1 - \rho}
\end{equation}

For the user activities in other states, we know that since not all tokens are user the activity level is optimal. Therefore it must be that:
\[
\frac{\underline{u}'(\bar{b}_t)}{(1 + \bar{\theta}_t)\bar{p}_t} = 
\frac{\underline{u}'(\bar{a}_t)}{\underline{p}_t} = 1,
\]
which actually match the first best optimum for their respective shock states. However, the conditions \Cref{eq:good_tax1} and \Cref{eq:good_tax2} for the remaining two states are bounded away from 1. If we find a choice of $\bar{\theta}_t$ such that they both get closer to 1, then we will have shown an improved outcome through taxation and burning of fees.

We use the equilibrium definition to determine other conditions restricting our choices of $\bar{p}_t, \underline{p}_t$ and $\bar{\theta}_t$:
\begin{itemize}
    \item In the high congestion case the chain is fully utilized. We have $\lambda \bar{a}_{t} + (1-\lambda) \bar{b}_{t} = 1$ and $\bar{u}'(\bar{a}_t) = \underline{u}'(\bar{b}_t) \ge \bar{p}_t \ge c'(1)$. Notice that $\bar{p}_t$ may be able to take more than one value, increasing $\bar{\theta}_t$ appropriately to control demand.
    \item The low-congestion state has the more restrictive condition $\underline{p}_t = c'(\lambda \underline{a}_t + (1-\lambda)\underline{b}_t)$.
\end{itemize}
We will try to tweak $\bar{p}_t$ and $\bar{\theta}_t$ so that $\bar{a}_t$ and $\bar{b}_t$ stay the same, but $\underline{b}_t$ and $\underline{a}_t$ increase.

Before doing so, it is useful to obtain an expression for $\bar{r}_t^T$ involving the remaining quantities. We have that:
\[
    \bar{r}_t^T = 
    \frac{\bar{\theta}_t \bar{p}_t(\lambda \bar{a}_t + (1-\lambda)\bar{\beta})t}{\lambda \bar{m}_{t-1} + (1-\lambda)\underline{m}_{t-1}}
    =
    \frac{\bar{\theta}_t \bar{p}_t(\lambda \bar{a}_t + (1-\lambda)\bar{\beta})t}{
    \lambda \frac{\bar{a}_t\bar{p}_t(1+\bar{\theta}_t)}{1 + \bar{r}_t^T} + (1-\lambda)\underline{p}_t \underline{b}_t}
\]
Observe that this is an increasing function of $\bar{\theta}_t$ and crucially, does not cancel out in \Cref{eq:good_tax1}. Therefore, there are multiple combinations of $\bar{\theta}_t$ and $\bar{p}_t$ leading to a congested equilibrium. We compare such $\bar{\theta}_t$ and $\bar{p}_t$ against an alternative without taxation and just $\tilde{p}_t$.

First off, notice that to have a congested equilibrium it needs to be $\tilde{p}_t < (1+\bar{\theta}_t) \bar{p}_t$, since the increase of $\bar{r}_t^T$ decreases the right hand side of \Cref{eq:good_tax1}. In addition (calling $\tilde{a}$ and $\tilde{b}$ the activities for $\tilde{p}_t$, we have that:
\[
\bar{u}'(\bar{a}_t) > \underline{u}'(\bar{b}_t) \text{ and } 
\bar{u}'(\tilde{a}_t) > \underline{u}'(\tilde{b}_t),
\]
since the right hand side of \Cref{eq:good_tax1} is greater than 1. In addition, $\tilde{b}_t > \bar{b}_t$ and $\tilde{a}_t < \bar{a}_t$. But then, taxation actually increases welfare in this state as well! By \Cref{lemma:poor_mans_convex}, we have that:
\[
\lambda \bar{u}(\bar{a}_t) + (1 - \lambda) \underline{u}(\bar{b}_t) > \lambda \bar{u}(\tilde{a}_t) + (1 - \lambda) \underline{u}(\tilde{b}_t).
\]
In the low congested case the result is immediate: the right hand side of \Cref{eq:good_tax2} approaches 1 due to $\bar{r}_t^T$, thus increasing the activity $\underline{b}_t$ without affecting $\underline{a}_t$, thus concluding the proof.
\end{proof}
\begin{remark}
Finally, it seems that the leeway in taxing without hurting consumption offered in the highly congested state can be taken advantage of, in order to transfer some of the surpluses (that would otherwise represent a monetary transfer to the validators) to other users, reducing their carry costs in periods of uncertainty and promoting their consumption. In addition, it seems that it is relatively easy to observe if taxation has been set properly: if the system remains congested and enough validators are motivated to process transactions, it is likely that the effects of this monetary policy are overall positive.
\end{remark}

\section{Conclusion}
We have developed a theory of token issuance policies in an infinite horizon model of a token economy with trading frictions. This is an important first step towards understanding optimal token issuance policies. However, several  directions for future work should be explored:
\subsection{Speculation and Bubbles}
Our model, assumes that users and validators, the main stakeholders in the ecosystem, are the only users of tokens. In reality a lot of token demand is driven by speculative activity by outside speculators. A natural extension of our model would incorporate speculation. A possible effect of speculation is the occurence of equilibria where the expected return of token holdings are temporarily much larger than than the risk free rate. In these circumstances, we conjecture that a more aggressive expansion of the token supply could be optimal policy.
\subsection{Continuous Time}
Our model considers discrete time that allows us to split time periods into two sub periods where the centralized token market on the one hand the transaction market on the other hand operate. The assumption of discrete time was helpful in setting up the model and introducing frictions in a natural way. However, it would also be instructive for analytical tractability to consider the continuous time limit of our model:
In the continuous time limit the carry cost can be modeled by a reduced form flow cost of holding money balances. Consumption shocks turn the optimal consumption problem over time into a stochastic control problem. This would lead us to a model similar to the model of~\cite{Cong2021}.
\subsection{Demand Estimation and Implementation}
The optimal taxation policy in our model assumes that willingness to pay and congestion can be identified when setting fees. The system is aware of the realization of the preference shock and sets fees accordingly. In reality, fees need to be set without knowing the willingness to pay of users. In particular, demand spikes need to be anticipated based on available data. Thus, the system needs to estimate demand based on observed demand in the past. Existing systems, such as Ethereum use step-wise myopic updating of the basefee based on the demand in the previous block(s). It is an open question, whether this kind of fee updating policy (or, equivalently, demand estimation) is close to the optimum or an alternative policy should be devised.  Alternatively, an auction mechanism could be deployed, in order to elicit willingness to pay directly from the users.

\bibliographystyle{plain}
\bibliography{main}

\begin{thebibliography}{10}

\bibitem{bitcoin-standard}
Saifedean Ammous.
\newblock {\em The Bitcoin Standard: The Decentralized Alternative to Central
  Banking}.
\newblock Wiley Publishing, 1st edition, 2018.

\bibitem{budish2022}
Eric~B Budish.
\newblock The economic limits of bitcoin and anonymous, decentralized trust on
  the blockchain.
\newblock {\em University of Chicago, Becker Friedman Institute for Economics
  Working Paper}, 2022.

\bibitem{chung2023foundations}
Hao Chung and Elaine Shi.
\newblock Foundations of transaction fee mechanism design.
\newblock In {\em Proceedings of the 2023 Annual ACM-SIAM Symposium on Discrete
  Algorithms (SODA)}, pages 3856--3899. SIAM, 2023.

\bibitem{CoKo1998}
Harold~L. Cole and Narayana~R. Kocherlakota.
\newblock {Zero nominal interest rates: why they're good and how to get them}.
\newblock {\em Quarterly Review}, 22(Spr):2--10, 1998.

\bibitem{Cong2021}
Lin~William Cong, Ye~Li, and Neng Wang.
\newblock Tokenomics: Dynamic adoption and valuation.
\newblock {\em The Review of Financial Studies}, 34(3):1105--1155, 2021.

\bibitem{CONG2022}
Lin~William Cong, Ye~Li, and Neng Wang.
\newblock Token-based platform finance.
\newblock {\em Journal of Financial Economics}, 144(3):972--991, 2022.

\bibitem{ferreira2021dynamic}
Matheus~VX Ferreira, Daniel~J Moroz, David~C Parkes, and Mitchell Stern.
\newblock Dynamic posted-price mechanisms for the blockchain transaction-fee
  market.
\newblock In {\em Proceedings of the 3rd ACM Conference on Advances in
  Financial Technologies}, pages 86--99, 2021.

\bibitem{friedman69}
Milton Friedman.
\newblock {The Optimum Quantity of Money}.
\newblock Macmillan, 1969.

\bibitem{garratt2018}
Rodney Garratt and Neil Wallace.
\newblock Bitcoin 1, bitcoin 2,....: An experiment in privately issued outside
  monies.
\newblock {\em Economic Inquiry}, 56(3):1887--1897, 2018.

\bibitem{Ireland2002}
Peter~N. Ireland.
\newblock {Implementing the Friedman Rule}.
\newblock NBER Working Papers 8821, National Bureau of Economic Research, Inc,
  March 2002.

\bibitem{kareken1981}
John Kareken and Neil Wallace.
\newblock On the indeterminacy of equilibrium exchange rates.
\newblock {\em The Quarterly Journal of Economics}, 96(2):207--222, 1981.

\bibitem{kiayias2023tiered}
Aggelos Kiayias, Elias Koutsoupias, Philip Lazos, and Giorgos Panagiotakos.
\newblock Tiered mechanisms for blockchain transaction fees.
\newblock {\em arXiv preprint arXiv:2304.06014}, 2023.

\bibitem{kocherlakota1998}
Narayana~R Kocherlakota.
\newblock Money is memory.
\newblock {\em journal of economic theory}, 81(2):232--251, 1998.

\bibitem{lagos2005}
Ricardo Lagos and Randall Wright.
\newblock A unified framework for monetary theory and policy analysis.
\newblock {\em Journal of Political Economy}, 113(3):463--484, 2005.

\bibitem{leonardos2021dynamical}
Stefanos Leonardos, Barnab{\'e} Monnot, Dani{\"e}l Reijsbergen, Efstratios
  Skoulakis, and Georgios Piliouras.
\newblock Dynamical analysis of the eip-1559 ethereum fee market.
\newblock In {\em Proceedings of the 3rd ACM Conference on Advances in
  Financial Technologies}, pages 114--126, 2021.

\bibitem{mayer2021}
Simon Mayer.
\newblock Token-based platforms and speculators.
\newblock {\em Available at SSRN 3471977}, 2021.

\bibitem{nakamoto2008}
Satoshi Nakamoto.
\newblock Bitcoin: A peer-to-peer electronic cash system, 2008.

\bibitem{textbook}
Ed~Nosal and Guillaume Rocheteau.
\newblock {\em Money, payments, and liquidity}.
\newblock MIT press, 2011.

\bibitem{pagnotta2022}
Emiliano~S Pagnotta.
\newblock Decentralizing money: Bitcoin prices and blockchain security.
\newblock {\em The Review of Financial Studies}, 35(2):866--907, 2022.

\bibitem{prat2021}
Julien Prat, Vincent Danos, and Stefania Marcassa.
\newblock Fundamental pricing of utility tokens.
\newblock 2021.

\bibitem{reijsbergen2021transaction}
Dani{\"e}l Reijsbergen, Shyam Sridhar, Barnab{\'e} Monnot, Stefanos Leonardos,
  Stratis Skoulakis, and Georgios Piliouras.
\newblock Transaction fees on a honeymoon: Ethereum's eip-1559 one month later.
\newblock In {\em 2021 IEEE International Conference on Blockchain
  (Blockchain)}, pages 196--204. IEEE, 2021.

\bibitem{roughgarden2020transaction}
Tim Roughgarden.
\newblock Transaction fee mechanism design for the ethereum blockchain: An
  economic analysis of eip-1559.
\newblock {\em arXiv preprint arXiv:2012.00854}, 2020.

\bibitem{schilling2019}
Linda Schilling and Harald Uhlig.
\newblock Some simple bitcoin economics.
\newblock {\em Journal of Monetary Economics}, 106:16--26, 2019.

\bibitem{Wilson1979}
Charles Wilson.
\newblock An infinite horizon model with money.
\newblock In {\em General Equilibrium, Growth, and Trade}, pages 79--104.
  Elsevier, 1979.

\bibitem{DBLP:conf/ccs/ZhangMMGJ20}
Fan Zhang, Deepak Maram, Harjasleen Malvai, Steven Goldfeder, and Ari Juels.
\newblock {DECO:} liberating web data using decentralized oracles for {TLS}.
\newblock In Jay Ligatti, Xinming Ou, Jonathan Katz, and Giovanni Vigna,
  editors, {\em {CCS} '20: 2020 {ACM} {SIGSAC} Conference on Computer and
  Communications Security, Virtual Event, USA, November 9-13, 2020}, pages
  1919--1938. {ACM}, 2020.

\end{thebibliography}

\end{document}